\documentclass[reqno]{amsart}

\usepackage[T2A]{fontenc}
\usepackage[utf8]{inputenc}
\usepackage[english]{babel}

\usepackage{amsfonts,amssymb,amsmath,mathabx,amsthm,fge}

\usepackage{graphics}

\newcommand{\op}{{\mathcal L}}
\newcommand{\opC}{\op_C}
\newcommand{\famC}{\mbox{\sf C}}
\newcommand{\llll}{{\ell}}

\newcommand{\somEE}{    E    }
\newcommand{\EEpm}[1]{\somEE_{\{\! #1 \!\}}\!}
\newcommand{\dEEpm}[1]{\somEE'_{\{\! #1 \!\}}\!}

\newtheorem{theorem}{Theorem}

\newtheorem{corollary}[theorem]{Corollary}
\newtheorem{lemma}[theorem]{Lemma}

\newtheorem{remark}{Remark}
\newtheorem{definition}{Definition}

\newcommand{\dcHe}{{\sf DCHE}}
\newcommand{\sdcHe}{{\sf sDCHE}}
\newcommand{\analytic}{{holomorphic}}
\newcommand{\Eqs}{{Eq.s}}
\newcommand{\backSlash}{\backslash}
\renewcommand*{\backSlash}{\,\fgebackslash\,}
\newcommand{\Imi}{\mathrm{i}}
\renewcommand{\Im}{\mbox{\sf Im}\,}
\renewcommand{\Re}{\mbox{\sf Re}\,}
\newcommand{\eiphi}{{\Phi}}
\newcommand{\eP}{{\Psi}}

\newcommand{\du}[1]{\tilde{#1}}
\newcommand{\deLta}{\Lambda}
\newcommand{\GamMa}{\Upsilon}

\newcommand{\puncturedS}{^\backprime{\!}S^1}

\newcommand{\diag}[1]{\,{\mathrm{diag}}(#1)}
\newcommand{\coveredone}{1}

\newcommand{\coveredmcone}{\mbox{\hspace{-0.1em}%
\raisebox{0.24em}{%
$\hspace{0.1em}
{
\hspace{-0.6em}
    {\scalebox{0.8}[0.6]{\rotatebox{+5}{$\curvearrowleft$}}}
\atop
\raisebox{0.1em}[1.pt][1.0pt]{\scalebox{1.0}{-1} }
}
  $
}\hspace{-1.em}
}}

\newcommand{\coveredmpone}{\mbox{\hspace{-0.4em}%
\raisebox{0.45em}{
$
{
\mbox{\scalebox{0.8}[0.6]{\rotatebox{180}{$\curvearrowright$}}}
\atop
 \scalebox{1.0}{-1}
}
  $%
}\hspace{-0.5em}
}}

\begin{document}

\title{
{
Interrelation of the Equation of RSJ Model of Josephson Junction and the Special Double Confluent Heun Equation}}
\author{ 
S.I. Tertychniy\\
~
\\ \small
VNIIFTRI, Russia
                                 }
\date{}

\maketitle

\pagestyle{plain}

\setcounter{page}{1}

\begin{abstract}
An explicit representation of
the maps
interconnecting
the sets of solutions
to the special double confluent Heun equation
and the equation of the RSJ model of overdamped Josephson
junction in case of shifted sinusoidal bias
is given.
The approach
leans on
specific properties
of
eigenfunctions of
a remarkable
linear operator acting on functions holomorphic on
the universal cover of the punctured complex plane.
The
functional equation the eigenfunctions noted obey is derived.
The matrix form of the monodromy transformation
they manifest
 is given.
\end{abstract}

%
%
%


\section*{\protect\centering{\sc Introduction }}\label{s:010}
In a sense,
the non-linear first-order ordinary differential equation
\begin{equation}
   \label{eq:010}
\dot\varphi(t) +\sin\varphi(t)=B+A\cos\omega t,
\end{equation}
where the symbols
$A,B,\omega$ denote real constants,
stands out in
the dispersed totality
of
particular instances of
differential
equations
due to its emerging in
a number of problems of
physics,
mechanics, dynamical systems theory,
geometry \cite{Foo,FLT,GI}.
Perhaps most frequently this equation and its generalizations
appear in
investigations concerning
with
theoretical
study of dynamics of Josephson junctions  \cite{Ba,MK}.
Eq.~\eqref{eq:010} seems to be the most simple equation
(or, at least, should be considered
among the most simple ones)
which is able to properly
embody the so called phase lock effect
utilized in
many
devices
built upon
capabilities of
the Josephson effect.
The latter was theoretically predicted in 1962 and
was recognized in an experiment reported in 1963.
Thereafter, in 1968,
a heuristic model
of behavior
of a Josephson junction
incorporated in a circuit with given properties was proposed
\cite{St,McC}
which  
is currently referred to as RSJ
(or sometimes as RCSJ) model.
Eq.~\eqref{eq:010} follows
from it in   the limiting case of a small effective junction capacitance
under conditions when
its
effect
is negligible.
Besides, the right-hand side of \eqref{eq:010}
corresponds to excitation  (``bias'')
of a Josephson junction
by a
controllable DC (described by the dimensionless parameter $B$)
combined with an also controllable sinusoidal AC
of the dimensionless frequency $\omega$,
of the fixed (zero) initial phase,
 and of the given amplitude
(characterized by the dimensionless parameter $A$).

Eq.~\eqref{eq:010} suits well for
a fast and accurate numerical integration
and is thus convenient for application in numerical simulations.
At the same time, perhaps somewhat surprisingly,
``the pure mathematics'' associated
with it proves
to be fairly profound and is definitely of considerable interest.
Several approaches can be here employed while the most efficient one
starts with an appropriate complexifcation of the equation in question.
We consider below this step in details and establish equivalence
(mentioned for the first time in Ref.~\cite{T2})
of Eq.~\eqref{eq:010} to
a 
double confluent Heun equation.
The latter,
 in turn,
can be further explored by the methods of complex analysis
and the theory of linear differential equations in the complex domain.

\section*{\protect\centering{\sc  Transition of Eq.~\eqref{eq:010} to Complex Domain }}\label{s:020}

To begin with,
let us notice that
the right-hand side of \eqref{eq:010}
does not depend on $\varphi$ and
is periodic in the free real variable $t$.
We embody
this periodicity in the circular motion
coupled
to the varying real $t$
in the complex plane $\mathbb C$ around zero.
In other words, denoting a generic point in
$\mathbb C$ as $z$
we associate
$
z=e^{\Imi\omega t}
$
to processes described by the function $\varphi$
obeying Eq.~\eqref{eq:010}.

Similarly, the left-hand side of
the equation,
involving all the entries of $ \varphi $,
is invariant with respect to the shifts $\varphi\leftleftharpoons \varphi\pm2\pi$.
It suggest us to utilize in complex domain 
the exponent $\eiphi = e^{\Imi\varphi}$ instead of the original $\varphi$.
The next natural step is to 
consider $\eiphi$ as a holomorphic function of
$z$
and assume that the above equality takes place {\it on the unit circle},
i.e.\ when $ z=e^{\Imi\omega t}$. In other words, it is assumed that
for real $t$ it holds
$ \eiphi( e^{\Imi\omega t})= e^{\Imi\varphi(t)} $.
At the same time,
for generic $z$,
the function
$\eiphi(z)$ becomes
the analytic continuation
of its instantiation
on
the above circle,
being therefore not pointwise representable
through
values of 
$\varphi$.
The
next and
also
last action
in the constructing of the transformation
we search for
is the selecting
of a differential equation constraining holomorphic $\eiphi(z)$
in
such a way that,
when restricted to the unit circle,
it would turn into
Eq.~\eqref{eq:010}, provided the above identifications
are taken into account. Such an equation can easily be found. 
It reads
\begin{equation}%
\label{eq:020}
z^2 {\eiphi}'
=
(2\Imi \omega)^{-1}z\,(1 - {\eiphi}^2)+
\big(\llll\,z+\mu(z^2+1)\big){\eiphi}.
\end{equation}
Here $\llll=B/\omega, \mu=A/2/\omega$ are the new but cognate
constant parameters.

The  non-linear ODE   \eqref{eq:020}      belongs to the 
Riccati' family. It is well known
that all these equations
are convertible to
certain linear second-order ODEs. 
We are going to employ such an equivalence
leading, in our case, to a double confluent Heun equation.
However, we
make here use of an
 indirect method for its derivation
which is
based on
inspection of properties of a remarkable linear operator $\opC$
defined in the next section and having, at first glance, no relation
  to the equation of RSJ model.

\section*{\protect\centering{\sc  Operator $\opC$  }}\label{s:030}

Let us consider the linear
operator 
$\opC$ which sends
a
holomorphic function $E$
of the complex argument $z$
to the function $\opC[\somEE]$ of the same argument
as follows
\begin{eqnarray}
     \label{eq:030}
\opC: \somEE(z)\mapsto
\opC[\somEE\,](z)
\!&=&\!
2\,\omega \,
z^{-\llll-1}
\raisebox{-0.8ex}[0.8ex][0.9ex]{\bigg\lfloor}%
_{\mbox{
      \rlap{\hspace{-1.2ex}
$\genfrac{}{}{0pt}{1}{
                     \phantom{.|}
                     }{
                      z\leftleftharpoons{}z^{-1}
                       }
$
           }
        }
  }
\hspace{-1.3ex}
\big( \somEE'(z) - \mu \somEE(z)\big).
\end{eqnarray}
Here
the symbols $\llll,\mu, \omega$
denote the
 constant parameters
which are, for now,
considered
arbitrary except
for the
claim of fulfillment of the
non-degeneracy conditions
$\mu\not=0\not=\omega$.
The mark $\lfloor_{z\leftleftharpoons{}z^{-1} } $
indicates the operation of
replacement of the variable in the expression
situated to the right of it.
The function
$\opC[\somEE](z)$
is obviously holomorphic
in the correspondingly transformed domain, provided the latter does not contain zero.

It is natural to adopt
as the domain  $\mathbf{\Omega}$ of
the operator
$\opC$
some set (say, a linear space)
of functions
which is preserved under %
its action.
This assumption
requires
 of
the domain $\Theta$
of the members
of $\mathbf{\Omega}$  
 to be invariant with
respect to the map $\famC: z\mapsto z^{-1}$ or,
at least, to produce
a nonempty intersection
 $\famC\,\Theta \cap  \Theta\not=\emptyset$.
The punctured complex plane
$\mathbb{C}^*= \mathbb{C}\backSlash 0$ is an example of such a domain.
For the sake of definiteness,
we shall utilize it, provisionally,
in the role of $\Theta$.
This
turns out to be
 not the best solution but
later on
we shall become able to specify
 a more appropriate $\Theta$ realization.

The following statement holds true.
\begin{lemma}\label{l:010}
The squared (composed with itself) transformation $\opC$
preserves its argument, i.e.\
\begin{equation}
\opC\circ\opC[\somEE]=\somEE,
\end{equation}
if and only if
the function $\somEE=\somEE(z)$
obeys the equation
\begin{eqnarray}
\label{eq:050}
&&z^2 \somEE''
+\big((\llll+1)
 z
+ \mu (1-z^2) \big) \somEE'
+
(  -\mu  (\llll+1) z
+
\lambda
%
)
\somEE
=0,
\\
\label{eq:060}
&&\mbox{where } \lambda=(2\omega)^{-2}-\mu^2.
\end{eqnarray}
\end{lemma}\noindent
\begin{proof}
The
above assertion
immediately follows from
the identity
\begin{eqnarray}
  \label{eq:070}
  \opC\circ\opC[\somEE](z)
  &\equiv&
\somEE(z)
  -
  (2\omega)^2\cdot\mathrm{lhs}\mbox{\eqref{eq:050}},
\end{eqnarray}
where
`$ \mathrm{lhs}\mbox{\eqref{eq:050}} $'
stands for the left-hand side expression of
Eq.~\eqref{eq:050}
considered as a function of $z$.
In particular, the equality  \eqref{eq:060}  also
follows from
 a straightforward computation
 verifying
Eq.~\eqref{eq:070} by means of expansion of its left-hand side.
\end{proof}

It has to be noted that
the ordinary second order
linear homogeneous differential equation \eqref{eq:050}
belongs to the family of
so called
double confluent Heun equations (often referred to as 
\dcHe{} or similarly).
They are
discussed
in  Ref.s~\cite{SW,SL};
see also the online resource 
\cite{HeunProject}
and Ref.~\cite{H}
 for
more
recent
bibliography.
A generic \dcHe{}
is identified
by
\textit{four\/}
constant parameters
while Eq.{~}\eqref{eq:050} involves
only  \textit{three} ones.
Accordingly,
it was
suggested to
name
 Eq.{~}\eqref{eq:050} the
\textit{special\/}
double confluent Heun equation (which may be referred as \sdcHe, accordingly)
and the term is
here
adopted
for definiteness
as well.

The clarification
of relationship
between
the equations \eqref{eq:050} and 
                                 \eqref{eq:010}
can be built upon the study of
the eigenfunctions of the
operator \eqref{eq:030}. The principal point
is here that any such eigenfunction
is automatically
a solution to Eq.{~}\eqref{eq:050} (for the appropriate value of the
parameter $\lambda$).
Indeed, the following  statements holds true:
\begin{lemma}\label{l:020}
An eigenfunction of the operator
$\opC$
with eigenvalue
$\nu\not=0$
obeys 
Eq.{~}\eqref{eq:050}
with
$\lambda=\nu^2/(2\omega)^2-\mu^2$.
If the parameter link
\eqref{eq:060}
is met 
then $\nu^2=1$.
\end{lemma}\noindent
\begin{proof}
The lemma assertion
 follows in obvious way from the same identity
\eqref{eq:070}.
\end{proof}
\begin{remark}\label{r:010}
\hangindent=2ex
\rm
Omitting above the case $\nu=0$,
we are not at risk to
forfeit
any non-trivial relationship
since, obviously,
the only
eigenfunction of $\opC$
with
null eigenvalue
is the exponent
$E(z)=\exp(\mu z)$
which can at any moment be taken into account, if necessary.
\end{remark}

Motivated by
the two above lemmas, we may
adopt
the set of solutions to
Eq.{~}\eqref{eq:050}, which constitutes
a 2-dimensional
linear space,
as
the functional space $\mathbf\Omega$
on which the action of the operator $\opC$
has to  
be considered.

\begin{remark}\label{r:020}
\hangindent=2ex
\rm
In  canonical representation
(i.e.\ when resolved with respect to the higher derivative)
the linear differential equation \eqref{eq:050}
suffers of the only singularity
situated
at $z=0$. Its solutions are thus holomorphic
everywhere except at zero.
The singularity
of the equation
at the center of $\mathbb C$
is irregular.
However 
solutions
holomorphic thereat may, in principle, exist.
This can occur
only on some special subset of ``tuned'' constant parameters
of lower dimension, see Ref.s~\cite{T3,BT3,BT2}.
Moreover, even on it,
only a single (unique up to a constant factor) solution
is regular at zero whereas all other ones are not.
It means that,
when considering the
common domain
$\Theta$
of functions constituting $\mathbf\Omega$,
one must remove
from it
the center $z=0$.
Then, starting from the complex plane, the
 punctured one
$\mathbb{C}^*= \mathbb{C}\backSlash 0$ arises.
It is not simply connected and as a consequence
a generic
solution $\somEE$ to Eq.{~}\eqref{eq:050},
excluding 
the mentioned exceptional cases of regularity at zero,
can not live on it.
The point is that
the analytic continuation of a generic solution 
to Eq.{~}\eqref{eq:050}
along non-homotopic
curves evading zero
may
produce different values at the point where they meet,
leading therefore to a
multi-valued function.
The non-uniqueness arises here since
the genuine
domain for solutions to Eq.{~}\eqref{eq:050}
is not a subset of $\mathbb{C}$ (such as $\mathbb{C}^*$)
but
a Riemann surface reducing here to
 {\it the universal
cover\/}
$\tilde{\mathbb{C}}\mathstrut^*$
of $\mathbb{C}^*$.
This surface 
is
diffeomorphic to
$\mathbb{C}$,
the covering projection
$\Pi:\tilde{\mathbb{C}}\mathstrut^*\simeq\mathbb{C}
\stackrel{\exp}{\mapsto}
\mathbb{C}^*$
being realized by the natural exponential function.

As a consequence, %
when lifted to $\tilde{\mathbb{C}}\mathstrut^*$,
 the map
\begin{equation}
           \label{eq:080} 
\famC: z\mapsto z^{-1}
\end{equation} 
involved in the transformation \eqref{eq:030}
in the form of replacement of the free variable
loses the uniqueness
of its ``implementation''. %
Indeed, 
 $\famC$ may now have only a single fixed point.
 Hence
one has either $\famC(1)=1 $ and $\famC(z)\not=z$ for all the other
points
$z$ of the $\somEE$ domain including the (lifts of) $-1$,
or it holds
$\famC(-1)=-1 $ and $\famC (z)\not=z$ otherwise.
There is therefore no unaffected point playing role of $-1$ in the former case
(the first $\famC$ ``implementation'' ) and similarly for $+1$
in the latter case (for the alternative ``implementation''  of the map $\famC$).

These subtleties go beyond the scope of
the present notes,
however. In order to focus on
 the principal points
of  the
relationship in question,
we restrict our consideration
to \textit{a subset\/} (subdomain) of the
genuine
domain of functions verifying Eq.{~}\eqref{eq:050}.
Namely, we consider it to be
the open set
obtained
from  $\mathbb{C}^*$
by a removal of the ray of negative reals,
$\Theta^*
=
\mathbb{C}^*\backSlash\mathbb{R}_{<0}$.
The resulting subdomain
 is simply connected
and any function holomorphic in it, including solutions to Eq.{~}\eqref{eq:050},
is single-valued.
Besides,
the behavior
of the transformation
 $\famC$
remains (locally) ``standard'' %
and claims no precautions,
all  this
 at a price
of the dropping out from consideration
the  value
$-1\not\in\Theta^*$
of the
argument $z$
as well as all the other negative real numbers. %
\end{remark}

Let us
consider now  the
properties of the operator
$\opC$
in more details
and show how they
enables one to
establish the explicit form of
the
relationship of
the equations \eqref{eq:050} and \eqref{eq:020}
of which the latter is
directly related, in turn, to Eq.~\eqref{eq:010}.

\section*{\protect\centering{\sc  Eigenfunctions of the Operator $\opC$
      and Their Properties  }}\label{s:040}

Let 
the equation~\eqref{eq:050} with fixed parameters $\llll, \lambda, \mu $
such that $\lambda+\mu^2\not=0$ be given.
Then one can resolve Eq.{~}\eqref{eq:060}
with respect to $\omega$
(the scaling parameter in \eqref{eq:030}), i.e.\ select it obeying the equation
\begin{equation*}
4\omega^2(\lambda+\mu^2)=1.
\end{equation*}
Given such $\omega$,
we define the %
operator $\opC$ by the formula \eqref{eq:030}
treated ``as it stands''
in vicinity of $z=1$
and  assume that it
acts on the linear space
$\mathbf\Omega$ of solutions to Eq.{~}\eqref{eq:050}.
In view of the
lemma \ref{l:020},
an eigenfunction of the operator $\opC$
belonging to $\mathbf\Omega$
may only
correspond to either the eigenvalue $+1$ or to the eigenvalue
$-1$.
We denote such eigenfunctions (if they exist) by the
the symbols $\EEpm{+}$ and $\EEpm{-}$, respectively.

The following simple but important statements hold true.
\begin{lemma}\label{l:030}
\hspace{1ex}
\\[-1em]
\begin{itemize} 
\item
If a solution $\somEE=\somEE(z)$ to Eq.{~}\eqref{eq:050}
is an eigenfunction of the operator $\opC$
then
it solves the Cauchy problem for this equation
posed at $z=1$
with the initial data obeying one of the two constraints
\begin{equation}
\label{eq:090}
\somEE'(1)=
(\pm(2\omega)^{-1}+\mu) \somEE(1) .
\end{equation}
These correspond to the eigenvalues $\pm1$, respectively.
\item
The eigenfunctions $\EEpm{\pm} $, if exist,
obey the functional equation
\begin{equation}
  \label{eq:100}
\hspace{-1.1em}
\EEpm{+} (z)\EEpm{-}(1/z)+\EEpm{-} (z)\EEpm{+} (1/z)
=
2\,e^{\mu(z+1/z-2)}\EEpm{+} (1)\EEpm{-}(1).
\end{equation}
\end{itemize}%
\end{lemma}
\begin{corollary} \label{c:010}
\hspace{1ex}
\\[-1ex]
\begin{itemize}
\item $\EEpm{\pm}(1)\not=0$ for any eigenfunction
of the operator $\opC$.
\item There may exist not more than two, up to constant factors, eigenfunctions
of the operator $\opC$; their eigenvalues are distinct and amount to $\pm1$.
\end{itemize}%
\end{corollary}\noindent
Accordingly, the two
eigenfunctions $\EEpm{\pm}$
are linearly independent and hence
provide
the basis of the linear space $\mathbf\Omega$ of solutions to Eq.{~}\eqref{eq:050}.
\begin{proof}[Lemma proof]
In accordance with lemma \ref{l:020}, each eigenfunction
of the operator $\opC$ verifies Eq.{~}\eqref{eq:050}.
Next, by definition,
the property
of being 
an eigenfunction of $\opC$
with the eigenvalue either $+1$ or $-1$
is equivalent to
the equalities
\begin{equation}
     \label{eq:110}
\dEEpm{\pm}(z)=\pm(2\omega)^{-1} z^{-\llll-1}   \EEpm{\pm}(1/z)
+\mu \EEpm{\pm}(z),
\end{equation}
respectively.
Evaluating them at $z=1$, one obtains \Eqs{~}\eqref{eq:090}.

Further, considering Eq.{~}\eqref{eq:100},
let us denote as $U=U(z)$
the difference of its left-  and right-hand sides.
Computing its derivative and eliminating
 the
derivatives $ \dEEpm{\pm} $ by means of the equations
\eqref{eq:110}, the equation  $U'=\mu\cdot(z-1/z)\cdot U$ arises.
Since $U(1)=0$, obviously, this linear homogeneous
first order ODE
forces $U$ to coincide with
its trivial null
solution implying $U(z)\equiv0$.
The lemma is proven.
\end{proof}
\begin{remark}\label{r:030}
Yet another quite predictable
constraint which the eigenfunctions $\EEpm{\pm}$ obey reads
\begin{equation*}
 \dEEpm{+}(z)\EEpm{-}(z)- \EEpm{+}(z)\dEEpm{-}(z)=
\omega^{-1}z^{-\llll-1}
e^{\mu(z+1/z-2)}\EEpm{+} (1)\EEpm{-}(1)
\end{equation*}
It follows from consideration of the Wronskian for
Eq.{~}\eqref{eq:050} which applies since
the eigenfunctions
$ \EEpm{\pm} $ verify the latter.
\end{remark}

\section*{\protect\centering{\sc 
Explicit Representations of Eigenfunctions of $\opC$
}}\label{s:050}

As it has been mentioned,
the
eigenfunctions
of the operator $\opC$ can be
utilized for description of the space of solutions to
Eq.{~}\eqref{eq:050}.
However,  we should show,
 at first,
that they do exist.
We remind that the
defining property of
the functions $\EEpm{\pm}$
is equivalent to the claim of fulfillment
of one of the equations \eqref{eq:110}.
The latter are however
not the ``classical'' ODEs
since an unknown function
involved therein
is invoked with
two distinct arguments.
Hence
the
corresponding
standard theorem
of existence of solutions of ODE is here not directly applicable
and
a separate proof of existence of eigenfunctions of the operator
$\opC$ has  to be given.
To that end, let us consider the following
\begin{lemma}\label{l:040}
Let the \analytic{} function $\eiphi=\eiphi(z)$
defined in a simply connected vicinity of the point $z=1$
obey the 
Riccati equation (cf. Eq.~\eqref{eq:020})
\begin{eqnarray}
          \label{eq:120}
&&
z {\eiphi}'
+(2\Imi \omega)^{-1}({\eiphi}^2-1)=
(\llll+\mu(z+z^{-1})){\eiphi},
\end{eqnarray}
and the \analytic{} function $\eP=\eP(z)$
obey the (subsidiary) decoupled linear homogeneous first order ODE
\begin{equation}
\label{eq:130}
2\Imi\omega z\eP'=( {\eiphi} + {\eiphi}^{-1})\eP.
\end{equation}
Let also
\begin{equation}
  \label{eq:140}
|\eiphi(1)|=1\mbox{ and }\eP(1)=1.
\end{equation}
Then the expressions
\begin{equation}
       \label{eq:150}
\begin{aligned}
\EEpm{\pm}(z)&
=
2^{-1}
e^{\mu(z+1/z-2)/2}
z^{-\llll/2}
\times\\&\hspace{2ex}
\left\lgroup
\frac{1\pm\Imi}{\sqrt2}
(\eP(z)\eiphi(z))^{1/2}
+
\frac{1\mp\Imi}{\sqrt2}
(\eP(1/z)/\eiphi(1/z))^{1/2}
\right\rgroup
\end{aligned}
\end{equation}
determine 
the two eigenfunctions of the operator $\opC$ with
eigenvalues $\pm1$, respectively, provided neither of
them is the identically zero function.
In the exceptional case pointed out,
another function from the pair \eqref{eq:150}
is still a proper (non-trivial) eigenfunction of $\opC$.
\end{lemma}
\begin{proof}[Proof outline]
To verify the asserted
property of a function $ \EEpm{\pm}(z)$
(where one among the two sign symbols has been chosen and fixed),
 one has to compute
its derivative and
to
examine %
the fulfillment of the corresponding equation among \Eqs{~}\eqref{eq:110}.
In our case,
utilizing
the equations \eqref{eq:120} and \eqref{eq:130},
the aforementioned derivative
is expressed
in
terms of products of the same functions $\eiphi$ and $\Psi$
with the same arguments $z$ and $1/z$
which are involved in the  definition \eqref{eq:150}.
Subsequent algebraic
simplification establishes
the identical vanishing of the coefficients in front
of all the remaining products of $\eP$ and $ \eiphi $.
\end{proof}

The existence of the functions $\eiphi$ and $\Psi$ in vicinity of the point
$z=1$ is ensured by the wellknown theorem of existence of local solution
of Cauchy problem for ordinary differential equations.
In case of \Eqs{~}\eqref{eq:120} and \eqref{eq:130},
one can  state 
even more according to the following
\begin{lemma}\label{l:050}
Let the parameters $\llll, \mu,  \omega $ be real and $\omega>0$.
In case of initial
conditions obeying the constraints
\eqref{eq:140},
the solution
$\eiphi(z)$,  $\Psi(z)$
of the Cauchy problem for
the system of
equations~\eqref{eq:120} and \eqref{eq:130}
exists
in some vicinity of the ``punctured unit circle''
\begin{equation}
\label{eq:155}
\puncturedS=\{z\in\mathbb{C}, |z|=1, z\not=-1\},
\end{equation} 
both functions  $\eiphi(z)$,  $\Psi(z)$
having also no zeros therein.
\end{lemma}
\begin{proof}
Let us restrict Eq.{~}\eqref{eq:120} to the unit circle embedded into $ \mathbb{C} $
and parameterized
by means of the
substitutions 
\begin{equation}
       \label{eq:160}
z \leftleftharpoons 
e^{\Imi\omega t},\; \eiphi(z)
\leftleftharpoons 
e^{\Imi\varphi(t)},\; t\in
\Xi
=
(-\pi\omega^{-1}, \pi\omega^{-1})
\subset\mathbb{R}.
\end{equation}
%
Then
we obtain exactly Eq.{~}\eqref{eq:010} with the parameters
\begin{equation}
   \label{eq:170}
A=2\omega\mu, \; B=\omega\llll.
\end{equation}
%
Similar conversion %
of Eq.{~}\eqref{eq:130}
leads to
the equation
\begin{equation}
 \label{eq:180}
\dot P(t) =\cos \varphi(t),
\nonumber
\end{equation}
where
the function
$ P(t)$ is related to the original unknown  $\Psi(z)$ through %
the equation
$$ e^{P(t)}=\Psi(e^{\Imi\omega t} ). $$

For
any real $A, B,$  and $ \omega $,
Eq.{~}\eqref{eq:010} is 
 solvable on any segment of the real
axis
for any
 real initial data $\varphi(t_0)=\varphi_0$
set up
at
any prescribed
real
 $t_0$. Moreover, the corresponding solution is
a real-analytic function. Accordingly, let some real $\varphi_0$
be fixed and let the real-analytic function $\varphi(t)$
verify Eq.{~}\eqref{eq:010} on the segment
$\Xi 
$,
obeying
the initial condition $ \varphi(0)=\varphi_0 $.
Let us also introduce
the real-analytic function
$P(t)=\int^t_0\cos\varphi(\tilde t)\, d \,\tilde t  $
on the same domain $\Xi $.

The analytic continuation of the map
($\mathbb{C}\supset\mathbb{R}\supset)\,
\Xi 
\ni t\mapsto e^{\Imi\omega t}
\in {\puncturedS}(\subset \mathbb{C}^* )
$
establishes the \analytic{} diffeomorphism of some vicinity
of the segment $\Xi$ to
a
vicinity of ``the punctured unit circle''
$ \puncturedS $  \eqref{eq:155},
the former being in smooth bijection with the latter.
The \analytic{}  functions $\eiphi$ and $\Psi  $ arising as
the induced pullbacks of
analytic continuations of
the real analytic functions
$ e^{\Imi\varphi( t)}  $ and $e^{P(t)}  $, respectively, verify
\Eqs{~}\eqref{eq:120},   \eqref{eq:130}.
By definition,
they have no zeros on $\puncturedS $;
moreover, $|\eiphi|=1$ whereas
 $\Psi$ is real and strictly positive
therein. Hence there exist no their zeros in some vicinity of
$\puncturedS $ as well.
 Besides, in accordance
 with definitions and the posing of the Cauchy problem for the function $\varphi$,
it holds
 \begin{equation}
    \label{eq:190}
 \eiphi(1)=e^{\Imi\varphi_0},\,\Psi(1)=1
 \end{equation}
(where $\varphi_0$ can be chosen
arbitrary real).  \Eqs{~}\eqref{eq:140} are thus also fulfilled.
 The lemma is proven.
\end{proof}
          \begin{remark}\label{r:040}
\hangindent=2ex
\rm
The  non-uniqueness
of the square root function
involved
in Eq.{~}\eqref{eq:150}
is to be
 eliminated
by
means of
 the assignment
to the
functions
$\eiphi^{1/2} $, $\eiphi^{-1/2} $, and
$\Psi^{1/2} $
(the pullbacks of) the 
analytic continuations of the functions
$\exp\frac{\Imi}{2}\varphi(t) $,
$\exp\frac{-\Imi}{2}\varphi(t) $,
and
$\exp \frac{1}{2}
          \int_0^t\cos\varphi(\tilde t)d \tilde t $,
respectively.
               \end{remark}

\begin{remark}\label{r:050}
\hangindent=2ex
\rm
The requirement of the above lemma
claiming of
the constant parameters to be real
is motivated by \Eqs{~}\eqref{eq:170}, in which
the constants $A, B, \omega$ are constrained by
their meaning
inferred from
physical or geometrical
problems in which Eq.{~}\eqref{eq:010} is utilized.
Similarly,
the variable $t$ is there
interpreted as a (rescaled dimensionless) time or length.
While maintaining
contact with applications,
we assume below the above reality conditions
to be fulfilled throughout.
At the
same time, it is worth noting
that the existence results
(and most formulas evading application of complex conjugation)
remain valid, at least,
for sufficiently small variations of the parameters
shifting them from the real axis to $\mathbb C$.
\end{remark}

We see that
any
solution to Eq.{~}\eqref{eq:010} generates a pair of
eigenfunctions of the operator $\opC$ which are defined by
\Eqs{~}\eqref{eq:150} in terms of the functions $\eiphi(z)$ and $\Psi(z)$
the above lemma operates with.
However, one of them (not both, though) may
prove to be identical zero.
To clarify conditions
of appearance of
such a ``pathology'', we need the following
property of the eigenfunctions of $\opC$. 
\begin{lemma}\label{l:060}
Let us define the sequence of pairs of
 functions $\{a_k(z),b_k(z)\}, \\
k=1,2,\dots$,
\analytic{} everywhere except zero,
by means of the following recurrent scheme:
\begin{eqnarray}
  \label{eq:200}
&&\;
a_1=\mu,\;b_1=\pm(2\omega)^{-1}z^{-\llll-1};
\\
  \label{eq:210}
&&
\begin{aligned}
\
&
a_{k+1}=\mu a_{k} \mp(2\omega)^{-1}z^{\llll-1} b_k+ a_k',
\\
&
b_{k+1}=\pm (2\omega)^{-1}z^{-\llll-1} a_k - \mu z^{-2}b_k + b_k'.
\end{aligned}
\end{eqnarray}
Let also the functions $\EEpm{\pm}$ obey the equations
$\opC \EEpm{\pm}=\pm \EEpm{\pm}$. Then their derivatives
admit the following representations:
\begin{equation}
    \label{eq:220}
\frac{d^k}{d z ^k}\EEpm{\pm}(z) = a_k(z)\EEpm{\pm}(z)+b_k(z)\EEpm{\pm}(1/z),\; k=1,2,\dots
\end{equation}
In particular, it holds
\begin{equation*}
\frac{d^k}{d z ^k}\EEpm{\pm}(1) = (a_k(1)+b_k(1))\EEpm{\pm}(1),\; k=1,2,\dots
\end{equation*}
\end{lemma}
\begin{proof}
Let us apply the mathematical induction.
The induction base, the case $k=1$, reduces to
the equality which, in
view of \eqref{eq:200},  is
equivalent just
to the
corresponding
 equation $\opC \EEpm{\pm}=\pm \EEpm{\pm}$
 fulfilled by construction.
Next, let us compute the derivative of the both sides of Eq.\
\eqref{eq:220} for some fixed $k$,
eliminating afterwards $\dEEpm{\pm} $ on the right by means of
Eq.{~}\eqref{eq:220} get with $k=1$, and eliminating
the derivatives
 $a_k', b_k'$
with the help of \Eqs{~}\eqref{eq:210}.
As it can be shown by a straightforward computation,
the result reduces to
the same
equation 
\eqref{eq:220} in which
the index $k$ is replaced by $k+1$.
The induction step has thus been carried out and the lemma proof is accomplished.
\end{proof}
\begin{corollary}  \label{c:020}
The function
$ \EEpm{\pm}(z)$
defined by Eq.{~}\eqref{eq:150}
is the identically zero function
 if and only if $ \EEpm{\pm}(1)=0 $.
\end{corollary}

We apply the corollary \ref{c:020} to clarification  of
the conditions
leading
to
identically zero function $\EEpm{\pm}$ defined by Eq.{~}\eqref{eq:150}.
Indeed, substituting therein $z=1$ 
and taking into account \Eqs{~}\eqref{eq:190}, one gets
\begin{equation*}
\mbox{$
\EEpm{\pm}(1)=\mp\sin\frac{1}{2}(\varphi_0\mp\pi/2).  
$}
\end{equation*}
Hence one of the functions $ \EEpm{+} $ and $ \EEpm{-} $
can, indeed, be identical zero and this
takes place %
if and
only if $\varphi_0=\pi/2\,(\hspace{-1.7ex}\mod \pi)$.
\begin{remark}\label{r:060}
\hangindent=2ex
\rm
~\\[-1em]
\begin{itemize}
\item
The varying of the initial value $\varphi_0=\varphi(0)$
of a solution to Eq.{~}\eqref{eq:010}
results in appearance of
some additional constant factors.
This is
the  only
distinction of the functions $\EEpm{\pm} $,
obtained by means of \Eqs{~}\eqref{eq:150},
from the ``fiducial'' ones  corresponding to, say,  $\varphi_0=0$.
Besides, with respect to the case $\varphi_0=0$,
the absolute values of these $\varphi(0)$-dependent
factors do not exceed $1$.
\item In case of the identical vanishing of one of the functions
$\EEpm{\pm} $, the corresponding sum in brackets in
 \Eqs{~}\eqref{eq:150} vanishes. Then the same sum but with the opposite
choice of the signs amounts to twice its first summand.
Accordingly, the following factorized representation
of the nontrivial eigenfunction $\EEpm{\cdot}$ still
produced by one of  \Eqs{~}\eqref{eq:150}
arises:
\begin{equation*}
 \EEpm{\cdot} \;\propto\;
(e^{\mu(z+1/z-2)}
z^{-\llll}
\eP(z)\eiphi(z))^{1/2}.
\end{equation*}
As we have mentioned, this situation occurs
if 
$\varphi_0=\pi/2
\,
(\hspace{-1.7ex}\mod\pi)  $.
\end{itemize}
\end{remark}

Resuming, we have our first key
\begin{theorem} \label{t:010}
Let a solution $\varphi(t)$ to the equation \eqref{eq:010}
on the segment \\  $\Xi=(-\pi\omega^{-1}, \pi\omega^{-1})$ be given.
Then the analytic continuations
of the functions
$\exp(\Imi 
          \varphi(t) )$ and
$\exp( 
          \int_0^t\cos\varphi(t)d t) $
from $\Xi$
to some
vicinity of $\Xi$
in $\mathbb{C}  $,
converted
by means of the transformation \eqref{eq:160}
to the functions $\eiphi(z)$ and $\Psi(z)$
holomorphic in the corresponding vicinity of the 
punctured circle
\eqref{eq:155},
determine
therein
the two solutions
$\EEpm{\pm}=\EEpm{\pm}(z)$
to Eq.{~}\eqref{eq:050}
by means of the  formulas \eqref{eq:150}.
The functions  $\EEpm{\pm}$
are linearly independent
unless one of them is the identically zero function
that takes place if and only if
either
$\varphi(0)=\pi/2\,(\hspace{-1.7ex}\mod 2\pi)$
(leading to  $\EEpm{+}(z)\equiv0$)
or
 $\varphi(0)=-\pi/2\,  (\hspace{-1.7ex}\mod 2\pi)$
(leading to $\EEpm{-}(z)\equiv0$, respectively).
In case of linear independence the
functions $\EEpm{\pm}$ constitute the
basis of the space  ${\mathbf\Omega}$ of solutions to Eq.{~}\eqref{eq:050}.
\\
The functions $\EEpm{\pm} $
are also
the eigenfunctions with eigenvalues $\pm1$,
respectively,
of the linear
operator $\opC$ defined by Eq.{~}\eqref{eq:030};
$\opC$ %
is, thus, represented
in the
basis $\{\EEpm{+},\EEpm{-}  \}  $
by the diagonal matrix
$\diag{1,-1}$.
The linear space
${\mathbf\Omega}$
is invariant with respect to the operator $\opC$
which
acts on it as an
involutive automorphism.
\end{theorem}
\begin{corollary} \label{c:030}
The eigenfunctions of the  operator $\opC$ with
eigenvalues $\pm1$
are exactly the non-trivial solutions to  Eq.{~}\eqref{eq:050}
which obey the initial data constraint \eqref{eq:090}.
\end{corollary}
\begin{remark}\label{r:070}
\hangindent=2ex 
\rm  
Since the operator
$\opC$ is involutive
any non-trivial solution to
Eq.{~}\eqref{eq:050}
is
either its eigenfunction itself
or the expressions $const\cdot(\somEE\pm\opC\somEE)$
constitute a pair of such eigenfunctions
which are linearly independent. Adjusting the above factor $const$,
they can be made real (self-conjugated, see the next section).
\end{remark}

\section*{\protect\centering{\sc 
 Self-Conjugation Property of Eigenfunctions of
the Operator  $\opC$
}}\label{s:060}

The explicit formulas for eigenfunctions of the operator
$\opC$ enables one an easy establishing of their invariance with respect
to the complex conjugation. However,
the
analogous relations
for the
functions $\eiphi$ and $\Psi$ involved
in $\EEpm{\pm}$ definition \eqref{eq:150}
have to be derived beforehand. %
To that end,
let us
introduce the following
auxiliary
working definition.
\begin{definition}\label{d:010}
Let $\GamMa(z)$ be any function holomorphic
in some connected
and simply connected
open subset of $ \mathbb C$ containing the point
$z=1$. We shall
name the function
\begin{equation}
     \label{eq:260}
\du{\GamMa}(z)
=
\overline{ {\GamMa}(1/\overline z)}
\end{equation}
{\sl dual\/} to the function $\GamMa(z)$.
\end{definition}
\begin{remark}\label{r:080}
\hangindent=2ex
\rm~
The above definition obviously
implies that
\begin{itemize} 
     \item The 
function dual to a holomorphic function
is also holomorphic in some open set containing the point $z=1$;
the intersection of the domains of ${\GamMa}$ and
$\du{\GamMa}$
is open, non-empty, and also contains $1$. 
       \item ``The
duality map''\; $\du{}:\GamMa\mapsto \du \GamMa  $ 
is 
involutive; in particular, %
the function ${\GamMa}(z)$ is, in turn, dual to the function
$\du{\GamMa}(z)$.
\end{itemize} 
\end{remark}

              \begin{lemma}\label{l:070}
Let the \analytic{} function $\eiphi=\eiphi(z)$
be a solution to
Eq.{~}\eqref{eq:120}
obeying the
constraint $|\eiphi(1)|=1  $ (cf.\ \Eqs{~}\eqref{eq:140}).
Then
\begin{equation}
  \label{eq:270}
\eiphi(z)\du\eiphi(z)=1.
\end{equation}
                  \end{lemma}\noindent
To prove the lemma,
we note first that
the function $\du{\eiphi}=\du{\eiphi}(z)$
dual to solution  ${\eiphi}(z)$ to Eq.{~}\eqref{eq:120}
obeys the equation
\begin{equation}
         \label{eq:280}
z \du{\eiphi}'
+(\Imi 2\omega)^{-1}(\du{\eiphi}^2-1)=-
(\llll+\mu(z+z^{-1})){\du\eiphi}.
\end{equation}
Then
a straightforward
computation shows that,
as a consequence of \eqref{eq:120} and  \eqref{eq:280}, it holds
\begin{equation}
          \label{eq:290}
\frac{d}{d z}({\eiphi}(z){\du\eiphi}(z)-1)
=(-2\Imi\omega z)^{-1}({\eiphi}(z)+{\du\eiphi}(z))({\eiphi}(z){\du\eiphi}(z)-1).
\end{equation}
Now let us introduce
an
auxiliary sequence of
functions $\delta_n$ (in fact, polynomials)
of the
three arguments
$z, \eiphi$, and $\du\eiphi$ which all are
regarded here, for a time, as free complex variables.
(It is worth noting that the functions
$\delta_n$ depends also on the parameters $\llll, \mu,\omega$
but these their arguments will be 
suppressed for the sake of the symbolism  simplicity.)
The functions $\delta_n$
are defined by means of the following recurrent scheme:
{\small\begin{eqnarray}
          \label{eq:300}
\delta_1&=&z(\eiphi+\du\eiphi),
\\
  \label{eq:310}
 \hspace{2em}
\delta_{n+1}&=&
(\eiphi+\du\eiphi+4\Imi\omega n  )\delta_n
-2 \Imi \omega z^2
\frac{\partial \delta_n}{\partial z}
\\&&
+(z(\eiphi^2-1) -2\Imi\omega(\llll z +\mu(z^2+1))\eiphi)
\frac{\partial \delta_n}{\partial \eiphi}
\nonumber\\&&
+(z(\du\eiphi^2-1) +2\Imi\omega(\llll z +\mu(z^2 +1))\du\eiphi)
\frac{\partial \delta_n}{\partial\du\eiphi}, n=1,2,\cdots
\nonumber
\end{eqnarray}}
We utilize them for introduction of
the functions
\begin{equation}
       \label{eq:320}
\deLta_n(z, \eiphi, \du\eiphi)
=
(-2\Imi\omega z^2)^{-n}\delta_n(z, \eiphi, \du\eiphi)
 (\eiphi\du\eiphi -1),\;n=1,2,\cdots.
\end{equation}
               \begin{lemma}\label{l:080} 
Under the conditions of the lemma  \ref{l:070},
it holds
\begin{equation}
     \label{eq:330}
\frac{d}{d z}\deLta_n(z, \eiphi(z), \du\eiphi(z))
=
\deLta_{n+1}(z, \eiphi(z), \du\eiphi(z)), n=1,2\cdots.
\end{equation}
                 \end{lemma}
\begin{proof}
It is easy to show that,
in view of
Eq.{~}\eqref{eq:120}
and Eq.{~}\eqref{eq:280},
the above assertion is
equivalent to
Eq.{~}\eqref{eq:300}
for $n=1$ and to
Eq.{~}\eqref{eq:310}
for $n>1$.
\end{proof}
                 \begin{corollary}\label{c:040}
Under the conditions of the lemma  \ref{l:070},
it holds
\begin{equation}
     \label{eq:340}
\frac{d^n}{d z^n}({\eiphi}(z){\du\eiphi}(z)-1)
=\deLta_n(z, \eiphi(z), \du\eiphi(z)), \; n=1,2,\cdots.
\end{equation}
             \end{corollary}
\begin{proof}
In case $n=1$ the above equation follows from \Eqs{}
\eqref{eq:290} and
\eqref{eq:300}, and the definition \eqref{eq:320}.
It is extended
to higher
derivative orders $n=2,3,\cdots  $
by
means of the mathematical induction
based on 
Eq.{~}\eqref{eq:330}.
\end{proof}
              \begin{corollary}\label{c:050}
Under the conditions
of the lemma \ref{l:070},
all the derivatives of the function $\eiphi(z)\du\eiphi(z) -1$ vanish at
the point
$z=1$.
          \end{corollary}
\begin{proof}
In accordance with $\deLta_n$ definition
\eqref{eq:320}
and Eq.{~}\eqref{eq:340},
for any $n=1,2,\cdots$
the derivative ${d^n}(\eiphi(z)\du\eiphi(z)-1)/{d z^n}  $
factorizes into a function holomorphic in vicinity
of the point $z=1$ times the function $\eiphi(z)\du\eiphi(z)-1 $ itself.
The latter is zero at the unity
(since $\eiphi(1)\du\eiphi(1)=|\eiphi(1)|^2=1$);
accordingly, the above multiple derivative 
is zero thereat as well. Thus all such derivatives
at $z=1$ are null.
\end{proof}
\begin{proof}[{Proof of the lemma  \ref{l:070} }]
Since the function $\eiphi(z)\du\eiphi(z)-1$ is analytic at the point $z=1$,
 the above
corollary  implies
its  identical vanishing and thus the
validity of the assertion of the  lemma \ref{l:070}.
\end{proof}

Similarly to above, let us consider
how the
function $\du\eP= \du\eP(z)$
dual to 
solution $ \eP = \eP(z)$ to
Eq.{~}\eqref{eq:130} is related to $ \eP $.
A straightforward
computation
establishes the fulfillment of
the equation
\begin{equation}
      \label{eq:350}
2 \Imi \omega z \du\eP'= ( \du\eiphi+\du\eiphi^{-1} )\du\eP.
\end{equation}
As a consequence, it holds
\begin{equation}
     \label{eq:360}
\frac{d}{d z}(\eP-\du\eP )=(4\Imi\omega z )^{-1}(\eiphi+\du\eiphi)(\eP-\du\eP ),
\end{equation}
provided
the functions  $ \eiphi=\eiphi(z)$  and $\du\eiphi= \du\eiphi(z)$
(mutually dual) obey Eq.{~}\eqref{eq:270}.
Let us notice
now that, as the  functions $ \eiphi$ and $ \du\eiphi $
are given, Eq.{~}\eqref{eq:360} can be regarded
as a 
linear homogeneous first order ODE 
for the holomorphic function $\delta=\delta(z)= \eP(z)-\du\eP(z)$ which is
correctly defined in the intersection of the domains of the functions
$ \eiphi$ and $\du\eiphi$ (with zero removed, if necessary).
As a consequence, one may claim that
the function
 $\delta$ either has no zeros in its domain or is the identically zero function.
But if the function $\eP(z)$ complies with ``the initial condition'' \eqref{eq:190}
then $\du\eP(1)=1$ as well implying $\delta(1)=0$. Thus $\delta(z)\equiv0$ at least
in a connected vicinity
of the point $z=1$.
We have therefore proven  the following
         \begin{theorem}\label{t:020}
Let the functions $\eiphi(z)$ and $ \eP(z)$ be holomorphic
in some connected
and simply connected
open subset of $\mathbb{C}^*$ containing the point
$z=1$,
obeying
therein the system of equations \eqref{eq:120}, \eqref{eq:130};
let the
constraints  \eqref{eq:140} be also fulfilled.
Then the equation
\eqref{eq:270}
and the equation
\begin{equation}
    \label{eq:370}
              \du\eP(z)=\eP(z)
\end{equation}
hold true.
           \end{theorem}
The lemma \ref{l:050} 
and the above theorem lead to the following
       \begin{corollary}\label{c:060}
Under the conditions of the theorem  \ref{t:020},
it holds
$|\eiphi|=1$ and $\Im \eP=0  $ on
``the punctured unit circle''  \eqref{eq:155}.
\end{corollary}
~\\[-3em]
\noindent
\begin{proof}
Since $\overline z= z^{-1} $ on the unit circle in $\mathbb C  $,
the assertions
to be proven
follow from \Eqs{~}\eqref{eq:270}
and \eqref{eq:370}.
\end{proof}
\begin{remark}\label{r:090}
\hangindent=2ex
\rm
We have shown, in particular, that
any 
 holomorphic functions $ \eiphi, \eP $ obeying
conditions 
of the theorem \ref{t:020}
determine 
the smooth real valued functions $\varphi(t), P(t)$
verifying the equations \eqref{eq:010} and \eqref{eq:180}, respectively.
\end{remark}

Now a short straightforward computation leaning on \Eqs{~}\eqref{eq:270}
and \eqref{eq:370} proves the following
         \begin{theorem}\label{t:030}
 Let the functions $\eiphi(z)$ and $ \eP(z)$
 obey the system of equations \eqref{eq:120}, \eqref{eq:130}
 and the
 constraints  \eqref{eq:140}.
Then
the functions
$ \EEpm{+}(z)$
and
$ \EEpm{-}(z)$
defined by \Eqs{~}\eqref{eq:150}
 are real
(self-conjugated), %
i.e.\ obey the constraints
\begin{equation}
         \label{eq:380}  
 \overline{ \EEpm{\pm}(\overline z)}=\EEpm{\pm}(z).
\end{equation}
          \end{theorem}

\section*{\protect\centering{\sc 
Representation of
general solution to the equation of RSJ model in terms of
solutions to special double confluent Heun equation
}}\label{s:070}

Having outlined the way of
constructing of solutions to Eq.{~}\eqref{eq:050} from
solutions to  Eq.{~}\eqref{eq:010},
 we
proceed with description of
the
inverse relationship.
It can be expressed  in the form of
the following
\begin{theorem} \label{t:040}
Let the \analytic{}
functions $\EEpm{+}(z)$ and $\EEpm{-}(z) $
be
the real (self-conjugated, see Eq.{~}\eqref{eq:380}) eigenfunctions
of the operator $\opC$ defined by Eq.{~}\eqref{eq:030}
with the corresponding eigenvalues $\pm1$;
let
also 
$\alpha$ be an arbitrary real constant. We define
the holomorphic
functions $\eiphi(z)$ and $\Theta(z)$
as follows:
{\small \begin{eqnarray}
\label{eq:390}
\hspace{2em}
\eiphi(z)
\!\!   &
        = 
       &\!\!
-\Imi  z^l
\frac{
\cos(\frac{1}{2}\alpha)\EEpm{+}(z)
+\Imi
\sin(\frac{1}{2}\alpha)\EEpm{-}(z)
}{
\cos(\frac{1}{2}\alpha)\EEpm{+}(1/z)
-\Imi
\sin(\frac{1}{2}\alpha)\EEpm{-}(1/z)
},
\\ 
\label{eq:400}
\Theta(z)
\!\!&
     = 
    &\!\!
-\Imi
\frac{
\cos(\frac{1}{2}\alpha)\EEpm{+}^2(1)\EEpm{-}(z)
+\Imi
\sin(\frac{1}{2}\alpha)\EEpm{-}^2(1)\EEpm{+}(z)
}{
\EEpm{+}(1)\EEpm{-}(1)
\big(
\cos(\frac{1}{2}\alpha)\EEpm{+}(z)
+\Imi
\sin(\frac{1}{2}\alpha)\EEpm{-}(z)
\big)
}.
\end{eqnarray}}
Then
\begin{itemize}
\item
the continuous function $\varphi(t)$ of the real variable $t$
determined
by
the equation
\begin{equation}
    \label{eq:410}
e^{\Imi\varphi(t)}=\eiphi(e^{\Imi\omega t})
\end{equation}
is well defined,
real valued, smooth  and verifying Eq.{~}\eqref{eq:010};
\item
the functions $P(t)$ and $Q(t)$
defined as follows
\begin{equation*}
P(t)
=
-\log(-\Im \Theta(e^{\Imi\omega t})), \;
Q(t)
= 
\Re\Theta(e^{\Imi\omega t})
\end{equation*}
are well defined,
real valued, smooth
and are related to %
the function
 $\varphi(t)$
by the subsequent quadratures
as follows
\begin{equation}
  \label{eq:430}
P(t)=\int_0^t  \cos\varphi(\tilde t)\,  d\,\tilde t,
\;
Q(t)= \int_0^t e^{-P(\tilde t)}  \sin\varphi(\tilde t) \, d \,\tilde t.
\end{equation}
\end{itemize}
\end{theorem}
\begin{remark}\label{r:100}
\hangindent=2ex
\rm
In view of the lemma \ref{l:040},
the both functions $\EEpm{\pm}(z)$
obey Eq.{~}\eqref{eq:050} and
one learns
from lemmas
\ref{l:040}
and
\ref{l:050}
that they always exist.
Hence, the functions $\eiphi(z)$ and $\Theta(z)$,
as well as the functions $\varphi(t), P(t), Q(t)$ which they give rise
to, %
are built
(and always can be built)
upon
solutions of this equation.
\end{remark}
\begin{proof}[Theorem proof]
Let us notice
that
since the functions $\EEpm{\pm}(z)$ obey
 a linear homogeneous second order differential equation
with coefficients \analytic{} everywhere except at zero (Eq.{~}\eqref{eq:050} times $z^{-2}$),
they are   themselves \analytic{}
 everywhere except, perhaps, at zero.
Besides,
in accord with the corollary \ref{c:020},
 $\EEpm{+}(1)\not=0\not= \EEpm{-}(1)$ that
eliminates 
the source of
an a priori conceivable
fault of the definition \eqref{eq:400}.

Now let us consider
the
identity
\begin{equation}
  \label{eq:440}
\begin{aligned}
&\Imi e^{\Imi\varphi(t)}
\big(\dot \varphi(t)+\sin\varphi(t)
-
\omega(\llll +2\mu\cos\omega t )
\big)   \equiv& \!\!\!
\big(e^{\Imi\varphi(t)}-\eiphi(e^{\Imi\omega t})\big)^{\scalebox{0.4}{\!\mbox{$\bullet$}}} 
\\ &&& \hspace{-27.4em}
+\left(
2^{-1}(e^{\Imi\varphi(t)}+\eiphi(e^{\Imi\omega t}))
-\Imi\omega (\llll^{\mathstrut}
+ 2 \mu \cos\omega t   )
\right)
\big(e^{\Imi\varphi(t)}-\eiphi(e^{\Imi\omega t})\big),
\end{aligned}
\end{equation}
which takes place for arbitrary smooth function $ \varphi(t) $
and which is proven by means of
straightforward computation
taking into account the
$\eiphi$ definition \eqref{eq:390} and
\Eqs{~}\eqref{eq:110}.
Thus it
follows from \eqref{eq:440}
that if
Eq.{~}\eqref{eq:410} is fulfilled then $ \varphi(t) $ verifies Eq.{~}\eqref{eq:010}
with $A=2\omega\mu, B=\omega \llll  $ (cf.\ \Eqs~\eqref{eq:170}).

Further,
let us note that
since the functions
$\EEpm{\pm}(z)$ are real,
one obtains
in case of a  real $\alpha$
 the following equalities:
$$
\overline{\eiphi(z)}=
\Imi  \bar z^l
\frac{
\cos(\frac{1}{2}\alpha)\EEpm{+}(\bar z)
-\Imi
\sin(\frac{1}{2}\alpha)\EEpm{-}(\bar z)
}{
\cos(\frac{1}{2}\alpha)\EEpm{+}(1/\bar z)
+\Imi
\sin(\frac{1}{2}\alpha)\EEpm{-}(1/\bar z)
}
\equiv \eiphi(1/\bar z)^{-1}.
$$
For $z= e^{\Imi\omega t} $ and real $t$, it holds $1/\bar z=z  $.
Accordingly,
one infers from  above that
$\overline{\eiphi(e^{\Imi\omega t} )} =\eiphi(e^{\Imi\omega t} )^{-1}$
and, consequently, $|{\eiphi(e^{\Imi\omega t} )}|=1$.
Then Eq.{~}\eqref{eq:410} yields $ |e^{\Imi\varphi(t)}| =1 $, and the
real-valued smooth function $\varphi(t)$ is determined
in terms of the logarithm of the non-zero smooth function
${\eiphi(e^{\Imi\omega t} )} $ in the standard way.
The first assertion of the theorem is therefore proven.

Addressing now the
second assertion, %
 let us introduce,
in addition to the function $\Theta(z)$, the function $\tilde\Theta(z)$
as follows:
\begin{equation}
\label{eq:460}
\tilde\Theta(z)
 = 
\Imi
\frac{
\cos(\frac{1}{2}\alpha)\EEpm{+}^2(1)\EEpm{-}(1/z)
-\Imi
\sin(\frac{1}{2}\alpha)\EEpm{-}^2(1)\EEpm{+}(1/z)
}{
\EEpm{+}(1)\EEpm{-}(1)
\big(
\cos(\frac{1}{2}\alpha)\EEpm{+}(1/z)
-\Imi
\sin(\frac{1}{2}\alpha)\EEpm{-}(1/z)
\big)
}.
\end{equation}

The functions $\Theta=\Theta(z) $ and $\tilde\Theta=\tilde\Theta(z) $
obey
the following system of the two linear homogeneous first order
differential
equations
\begin{equation}
  \label{eq:470}
\Imi\omega z \Theta'= -\eiphi^{-1}( \Theta - \tilde\Theta ),\;
\end{equation}
This is
the direct consequence of definitions and \Eqs{~}\eqref{eq:110}.

A straightforward verification also based on definitions
shows that for real eigenfunctions $ \EEpm{\pm} $
(and for real constant $\alpha$) it holds
$\overline{\Theta(z)}  = \tilde\Theta(1/\bar z)$,
i.e.\
the function $ \tilde\Theta $
defined by means of
a separate formula \eqref{eq:460}
is actually dual to the function $\Theta$ (see Eq.{~}\eqref{eq:260}).
As a consequence,
it holds
$\tilde\Theta(e^{\Imi\omega t} ) = \overline {\Theta(e^{\Imi\omega t}}) $. Then
\Eqs{~}\eqref{eq:470} yield the equation
\begin{equation*}
\frac{d}{d\,t}
\Theta(e^{\Imi\omega t})=-\eiphi(e^{\Imi\omega t})^{-1}
\big(
\Theta(e^{\Imi\omega t})-\overline{\Theta(e^{\Imi\omega t})}
\big).
\end{equation*}
Separating its
real and imaginary parts
and taking into account Eq.{~}\eqref{eq:410}, one gets
\begin{equation*}
\begin{aligned}
\frac{d}{d\,t} \Re\Theta(e^{\Imi\omega t})
=&-\Im \Theta(e^{\Imi\omega t})\sin\varphi(t),
\\ 
\frac{d}{d\,t} \Im\Theta(e^{\Imi\omega t})=&
-\Im \Theta(e^{\Imi\omega t})\cos\varphi(t).
\end{aligned}
\end{equation*}
In case of a given real valued function
$ \varphi(t) $,
the latter equation determining $ \Im\Theta $
can be
integrated by means of a quadrature.
Then the former one is integrated by means of another quadrature.
The integration constants are fixed
making use of
the initial
conditions $\Re \Theta(e^{\Imi\omega t})|_{t=0}=\Re\Theta(1)=0 $,
$\Im \Theta(e^{\Imi\omega t})|_{t=0}=\Im\Theta(1)=-1$
which 
follow from the
$\Theta$ definition \eqref{eq:400} evaluated at the point $z=1$.
The ultimate result of the integrations is just the formulas
\eqref{eq:430}. %
The theorem proof has been accomplished.
\end{proof}

Let us note that
for $t=0$ \Eqs{~}\eqref{eq:410} and \eqref{eq:390}
are equivalent to the equation
\begin{equation}
  \label{eq:500}
\mbox{$
 \EEpm{-}(1)\sin(\frac{1}{2}\varphi(0)-\frac{\pi}{4})\sin( \frac{1}{2}\alpha)
+
 \EEpm{+}(1)\cos(\frac{1}{2}\varphi(0)-\frac{\pi}{4})\cos( \frac{1}{2}\alpha)
=0.$}
\end{equation}
Obviously,
it is solvable
with respect to the angular parameter $\alpha$
for any given real $\varphi(0)$
(recall that the values of the
functions $ \EEpm{\pm} $  are real
when their argument is real and $ \EEpm{\pm}(1)\not=0$).
Conversely,
for any $\alpha \in[0,2\pi)$ some
 ``initial data''
 $\varphi_0=\varphi(0)\in[0,2\pi)$ obeying Eq.{~}\eqref{eq:500}
can be found.
We obtain, therefore, the following
        \begin{corollary}\label{c:070}
 \Eqs{}~\eqref{eq:410}, \eqref{eq:390}
enable one to obtain
any
solution to Eq.{~}\eqref{eq:010},
representing it in terms of solutions to Eq.{~}\eqref{eq:050}.
                \end{corollary}

\section*{\protect\centering{\sc Conclusion }}\label{s:080}

We have here shown
that any  solution to
the equation
\eqref{eq:010},
utilized
for the
modeling of dynamics of a Josephson junction,
can be converted to solutions
to Eq.{~}\eqref{eq:050} 
by means of a quadrature and analytic continuation of
two real analytic functions (theorem \ref{t:010}).
Moreover, in a generic case,
 a basis of the space of solutions to Eq.{~}\eqref{eq:050}
can then be
produced
and it is
constituted by the eigenfunctions 
of the operator $\opC $ (defined by Eq.{~}\eqref{eq:030});
moreover,
these are real (self-conjugated, see theorem \ref{t:030}).

Conversely, let  the two real eigenfunctions of the operator $\opC $
with eigenvalues $+1$ and $-1$
be given. Then all the solutions to Eq.{~}\eqref{eq:010}
can be obtained
making use of
the formulas 
\eqref{eq:390} and \eqref{eq:410} (theorem \ref{t:040},
corollary \ref{c:070}).
A similar formula, Eq.{~}\eqref{eq:400},
yields explicit representations of %
the integrals \eqref{eq:430} which are involved
in the criterion of the so called phase-lock \cite{T0,T2},
the remarkable property manifested under certain conditions by
solutions to Eq.{~}\eqref{eq:010} \cite{Ba,MK}. 

In total, the relationships
indicated above establish
the
explicit
1-to-1 correspondence between solutions spaces
of Eq.{~}\eqref{eq:050} and Eq.{~}\eqref{eq:010}, essentially,
because ambiguity still retained can be considered trivial.

It is also worth noting that the eigenfunctions
of the operator $\opC$ (as well as this operator on its own, of course)
are the important tools proving to be efficient in investigation of various problems
related to \sdcHe~\eqref{eq:050}. In particular,
the following explicit matrix representation $ \mathbf{M} $ of
{\em the monodromy
transformation}%
\footnote{
Here the monodromy transformation sends a solution to Eq.{~}\eqref{eq:050}
to another its solution which
is obtained from the former by means of
point-wise analytic continuation along 
counterclockwise oriented
full circle arcs
encircling
the singular center $z=0$.
On the set of solutions to Eq.{~}\eqref{eq:010},
the monodromy transformation
of solutions to Eq.{~}\eqref{eq:050}
is converted to the map $\varphi(t)\mapsto
{\mathcal M}\varphi(t)=
\varphi(t+2\pi/\omega)$.
Given $\mathbf{M}$, the
making use of  \Eqs~\eqref{eq:150},
\eqref{eq:390}, \eqref{eq:400}, etc enables one to obtain
an explicit representation of this transformation.
}
 of its space of solution
with respect to the basis $  \{\EEpm{+},\EEpm{-}\}$
can be obtained\footnote{
The formula \eqref{eq:510} had been derived
in case of integer orders $\llll$.
The cases of other   $\llll$
require additional examination.
}:
\begin{eqnarray}
\label{eq:510}
&&\hspace{2em}
\mathbf{M}=
 e^{4\mu}\big(2\EEpm{+}(\coveredone)\EEpm{-}(\coveredone) \big)^{-1}\times
\\&&
\hspace{-2.5em}
\begin{pmatrix}
\EEpm{+}(\coveredmcone)\EEpm{-}(\coveredmcone)
+
\EEpm{+}(\coveredmpone)\EEpm{-}(\coveredmpone)
&
\hspace{-3em}
\EEpm{+}(\coveredmcone)^2-\EEpm{+}(\coveredmpone)^2
\\
\EEpm{-}(\coveredmcone)^2-\EEpm{-}(\coveredmpone)^2
&
 \hspace{-3em}
\EEpm{+}(\coveredmcone)\EEpm{-}(\coveredmcone)
+
\EEpm{+}(\coveredmpone)\EEpm{-}(\coveredmpone)
\end{pmatrix}.
\nonumber
\end{eqnarray}
Here the symbols $\coveredmcone  $
and
$\coveredmpone  $ denote
the preimages of $-1\in\mathbb{C}^*$
in
the Riemann surface
$\tilde{\mathbb{C}}\mathstrut^*$,
the domain %
of generic
solutions to Eq.{~}\eqref{eq:050}, ``branching'' over $\mathbb{C}^*$
around ``the axis'' passing through the removed zero
(see the remark \ref{r:020}).
More exactly,
these
preimages of $-1$
are selected  as the closest ones
to the preimage of $1$, the fixed point
 of lifting of the transformation \eqref{eq:080}, which we may denote just as $1$.
Of them,
$\coveredmcone  $ is
reached from
$1$
along
an arc passed in the counterclockwise direction while for
 $\coveredmpone  $ similar arc is directed clockwise.
The above formula shows, in particular,
 that the diagonal elements of $ \mathbf{M} $
are real and coincide while the off-diagonal ones are pure imaginary.
It also follows from
 the equation
\eqref{eq:100}
that $\det \mathbf M=1$.
The two eigenvalues of the matrix \eqref{eq:510} coincide if and only if
one of its off-diagonal elements vanishes
and this observation
can be utilized as the base of
yet another
 criterion
of the phase-lock behavior
 for solutions to Eq.{~}\eqref{eq:010},
this time
referring to
properties
of  eigenfunctions
of the operator $\opC$.

\end{document}